\documentclass[10pt]{amsart}
\usepackage{amsmath,amssymb}

\begin{document}

\newtheorem{thm}{Theorem}[section]
\newtheorem{lem}[thm]{Lemma}
\newtheorem{prop}[thm]{Proposition}
\newtheorem{cor}[thm]{Corollary}
\newtheorem{defn}[thm]{Definition}
\newtheorem*{remark}{Remark}
\newtheorem{notn}[thm]{Notation}
\newtheorem{normal}[thm]{Normalization}

\numberwithin{equation}{section}

\newcommand{\Z}{{\mathbb Z}} %cph changed from \mathbf
\newcommand{\Q}{{\mathbb Q}}
\newcommand{\R}{{\mathbb R}}
\newcommand{\C}{{\mathbb C}}
\newcommand{\G}{{\mathbb G}}
\newcommand{\N}{{\mathbb N}}
\newcommand{\FF}{{\mathbb F}}
\newcommand{\T}{{\mathbb T}}
\newcommand{\A}{{\mathbb A}}
\newcommand{\fq}{\mathbb{F}_q}
\newcommand{\E}{{\mathbb E}}
\newcommand{\PP}{{\mathbb P}}

\def\scrA{{\mathcal A}}
\def\scrB{{\mathcal B}}
\def\scrC{{\mathcal C}}
\def\scrD{{\mathcal D}}
\def\scrE{{\mathcal E}}
\def\scrF{{\mathcal F}}
\def\scrH{{\mathcal H}}
\def\scrK{{\mathcal K}}
\def\scrL{{\mathcal L}}
\def\scrM{{\mathcal M}}
\def\scrN{{\mathcal N}}
\def\scrO{{\mathcal O}}
\def\scrPP{{\mathcal P}} 
\def\scrS{{\mathcal S}}
\def\scrT{{\mathcal T}}

\def\x{{\underline{x}}}
\def\X{{\underline{X}}}

\newcommand{\rmk}[1]{\footnote{{\bf Comment:} #1}}

\renewcommand{\mod}{\;\operatorname{mod}}
\newcommand{\ord}{\operatorname{ord}}
\newcommand{\TT}{\mathbb{T}}
\renewcommand{\i}{{\mathrm{i}}}
\renewcommand{\d}{{\mathrm{d}}}
\renewcommand{\^}{\widehat}
\newcommand{\HH}{\mathbb H}
\newcommand{\Vol}{\operatorname{vol}}
\newcommand{\area}{\operatorname{area}}
\newcommand{\tr}{\operatorname{tr}}
\newcommand{\norm}{\mathcal N} % norm =(\frac{ n+\sqrt{n^2-4}} 2)^2
\newcommand{\intinf}{\int_{-\infty}^\infty}
\newcommand{\ave}[1]{\left\langle#1\right\rangle} %  average
\newcommand{\Var}{\operatorname{Var}}
\newcommand{\Cov}{\operatorname{Cov}}
\newcommand{\Prob}{\operatorname{Prob}}
\newcommand{\sym}{\operatorname{Sym}}
\newcommand{\disc}{\operatorname{disc}}
\newcommand{\CA}{{\mathcal C}_A}
\newcommand{\cond}{\operatorname{cond}} % conductor
\newcommand{\lcm}{\operatorname{lcm}}
\newcommand{\Kl}{\operatorname{Kl}} %Kloosterman sum
\newcommand{\leg}[2]{\left( \frac{#1}{#2} \right)}  % Legendre symbol
\newcommand{\SL}{\operatorname{SL}}
\newcommand{\Id}{\operatorname{Id}}
\newcommand{\diag}{\operatorname{diag}}
\newcommand{\diam}{\operatorname{diam}}
\newcommand{\supp}{\operatorname{supp}}
\newcommand{\sgn}{\operatorname{sgn}}

\newcommand{\sumstar}{\sideset \and^{*} \to \sum}

\newcommand{\LL}{\mathcal L} %L-function of u
\newcommand{\sumf}{\sum^\flat}
\newcommand{\Hgev}{\mathcal H_{2g+2,q}}
\newcommand{\USp}{\operatorname{USp}}
\newcommand{\conv}{*}
\newcommand{\dist} {\operatorname{dist}}
\newcommand{\CF}{c_0} % Fejer constant
\newcommand{\kerp}{\mathcal K}

\newcommand{\gp}{\operatorname{gp}}
\newcommand{\Area}{\operatorname{Area}}

\title[Delocalization for random displacement models]{Delocalization for random displacement \\ models with Dirac masses}
\author{Henrik Uebersch\"ar}
\address{Laboratoire Paul Painlev\'e, CNRS U.M.R. 8524, Universit\'e Lille 1, 59655 Villeneuve d'Ascq Cedex, France.}
\email{henrik.ueberschar@math.univ-lille1.fr}
\date{\today}
\maketitle

\begin{abstract}
We study a random Schr\"odinger operator, the Laplacian with random Dirac delta potentials on a torus $\T_L^d=\R^d/L\Z^d$, in the thermodynamic limit $L\to\infty$, for dimension $d=2$. The potentials are located on a randomly distorted lattice $\Z^2+\omega$, where the displacements are i.i.d. random variables sampled from a compactly supported probability density. We prove that, if the disorder is sufficiently weak, there exists a certain energy threshold $E_0>0$ above which exponential localization of the eigenfunctions must break down. In fact we can rule out any decay faster than a certain polynomial one. Our results are obtained by translating the problem of the distribution of eigenfunctions of the random Schr\"odinger operator into a study of the spatial distribution of two point correlation densities of certain random superpositions of Green's functions and its relation with a lattice point problem.
\end{abstract}

\section{Introduction} 

Anderson observed in his landmark 1958 paper \cite{A} that in the presence of sufficiently strong disorder the wave functions of a disordered quantum system may be exponentially localized. This phenomenon is today known as ``Anderson localization''. A key question in the theory of disordered systems concerns the breakdown of localization and the existence of a transition from a localized to a delocalized regime at weak disorder. The mathematical theory of Schr\"odinger operators with random potentials on $\R^d$ is concerned with a rigorous mathematical understanding of Anderson localization. Whereas the exponential localization of the eigenfunctions of random Schr\"odinger operators is by now firmly established at the rigorous level, the problem of delocalization remains widely open.

The scaling theory of Abrahams, Anderson, Licciardello and Ramakrishnan \cite{AALR} predicts that the existence of a delocalization transition ought to depend on the dimension of the system. In dimension $d=1$ one always expects exponential localization of the eigenfunctions no matter how weak the disoder in the system. In dimension $d=3$ a phase transition from localization at strong disorder to delocalization at weak disorder is expected to occur.  

The critical case of dimension $d=2$ is of particular interest. Generally it is believed that, as for dimension $d=1$, any strength of disorder is sufficient to have exponential localization and no transition occurs. Although, the localization length is expected to be exponentially large in terms of the electronic mean free path length, which makes it very difficult in practice to distinguish the two regimes.

As we will show in this paper, the predictions of the scaling theory surprisingly do not hold for certain random Schr\"odinger operators for which we prove delocalization, i.e. the breakdown of exponential localization for sufficiently weak disorder above a certain energy threshold in dimension $2$. In fact we can rule out any decay but a certain polynomial one.

\subsection{The model}

The present paper studies a random displacement model on a large torus (more precisely, the limit as the size of the torus tends to infinity) in dimension $2$. The impurities are modeled by Dirac delta potentials, which is a natural simplification from smooth compactly supported potentials which offers the considerable advantage that the eigenfunctions may be computed explicitly as certain superpositions of Green's functions. The complexity of the eigenfunction is encoded in the superposition vector (cf. section \ref{back} for a detailed discussion). It is known that smooth potentials may be modeled by delta potentials in a suitable energy range, where the wavelength is much larger than the size of the support of the smooth potential (cf. for instance \cite{SebaExner,Fishman}).

We consider a weak disorder regime. This means given a torus $\T_L^2$, where $L\in\N$ is large, we introduce the following random displacement model on $\T^2_L$
\begin{equation}\label{random Schrodinger}
H_{\omega_L}=-\Delta+\alpha\sum_{\xi\in\Z^2\cap \T^2_L}\delta(x-\xi-\omega_\xi)
\end{equation}
where the displacements $\omega_\xi$ are i.i.d. random variables with a compactly supported radial probability density $P_{0}(x)=\epsilon_0^{-2}P(\frac{|x|}{\epsilon_0})$, where $\epsilon_0\in(0,\tfrac{1}{4})$ is the disorder parameter, $P\in C^0_c(\R_+)$ and $0\in\supp P\subset [0,1]$. We impose periodic boundary conditions, but our arguments also work for Neumann or Dirichlet boundary conditions.

We denote by $$\omega_L=\{\omega_\xi \mid \xi\in \Z^2\cap \T^2_L\}$$ the stochastic process that samples independently from the probability density $P_0$ the random displacements $\omega_\xi$ for each $\xi\in \Z^2\cap \T^2_L$. 

The above operator may be formulated rigorously by applying the theory of self-adjoint extensions (see subsection \ref{extensions}) to the restricted Laplacian $-\Delta|_{C^\infty_c(\T^2_L-\omega_L)}$. 
We denote the family of self-adjoint extensions associated with the formal operator \eqref{random Schrodinger} by $\{-\Delta_{\omega_L,U}\}_{U\in U(N)}$, where $N=\#\omega_L$. The number of self-adjoint extensions exceeds the number of physical coupling constants. We remark that in particular the subgroup of diagonal unitary matrices $D(N)\subset U(N)$ corresponds to the case where a non-local interaction between the impurities is forbidden.

Since the operator $-\Delta_{\omega_L,U}$ is a rank $N$ perturbation of the Laplacian, it has at most $N$ ``new" (random) eigenfunctions corresponding to each new eigenvalue which is ``torn off'' each old eigenspace of the Laplacian, and we remark that for $L\gg 1$ the rank is always larger than the multiplicity of the Laplace eigenvalues. This means that there will be no ``old'' Laplace eigenfunctions in the spectrum of  $-\Delta_{\omega_L,U}$.

The eigenvalues of the Laplacian on the torus $\T^2_L$ are given by the set $S=\{n \mid n=4\pi^2(\xi_1^2+\xi_2^2)/L^2, \;\xi_1,\xi_2\in\N\}=\{0=\frac{n_0}{L^2}<\frac{n_1}{L^2}<\frac{n_2}{L^2}<\cdots\}$, where the set $\scrS=\{n_k\}_{k=0}^\infty$ are integers which (up to a factor $4\pi^2$) are representable as sums of two squares of integers. The associated eigenfunctions are of the form $e_{\xi/L}(x)=e^{2\pi\i \left\langle\xi/L,x\right\rangle}$.

The multiplicity of a Laplace eigenvalue $n$ is given by the number of ways the integer $nL^2/4\pi^2$ can be written as a sum of $2$ squares of integers. The multiplicity of $n$ grows on average like $\sqrt{\log (nL)}$, which is a consequence of Landau's Theorem \cite{L}: $$\#\{n\in \scrS \mid n\leq x\}\sim \frac{Bx}{\sqrt{\log x}}$$ for some $B>0$. 

\subsection{Delocalization}

We are interested in the spatial distribution of the eigenfunctions of $H_{\omega_L}$ on the torus $\T^2_L$ in the limit as $L\to\infty$.

Let $\omega=\{\omega_\xi \mid \xi\in\Z^2\}$ denote the stochastic process which independently samples from $P_0$ the random displacements for each lattice vector $\xi\in\Z^2$.

The scaling theory of Abrahams, Anderson, Licciardello and Ramakrishnan \cite{AALR} suggests that the eigenfunctions of the formal random Schr\"odinger operator 
\begin{equation}\label{random R2}
H_\omega=-\Delta+\alpha\sum_{\omega_\xi\in\omega}\delta(x-\xi-\omega_\xi), \quad \alpha\in\R
\end{equation}
ought to be exponentially localized for all energies (since $d=2$). 

The $L^2$-eigenfunctions are random superpositions of Green's functions and are of the form %(cf. for instance \cite{BG,HKK}) 
$$\Psi_E(x)=\sum_{\omega_\xi\in\omega}c_{\xi,\omega}G_E(x,\xi+\omega_\xi), \quad (c_{\xi,\omega})_{\xi\in\Z^2}\in l^2.$$
The physical interpretation of this exponential localization (``Anderson localization'') is that transport breaks down due to the presence of sufficient disorder in the system. It has been shown for various models, that in the localized regime the random operator has almost surely pure point spectrum \cite{Molchanov}.

In the case $d=3$, however, the scaling theory predicts the existence of a so-called ``mobility edge'', which means that for sufficiently weak disorder, above a certain energy threshold, a continuous band structure should emerge in the spectrum of the random operator. Values well inside each interval will correspond to generalized eigenfunctions, whereas values near the band edges may still correspond to exponentially localized eigenfunctions. For very low disorder, and sufficiently high energy, the spectrum should be purely continuous (possibly with a singular component) and all values should correspond to generalized eigenfunctions. 

In fact almost sure existence of pure point spectrum and exponential localization of the eigenfunctions at the bottom of the spectrum has been proven by Boutet de Monvel and Grinhpun \cite{BG} for the case of random couplings and scatterers located on a lattice. This was later extended to the case of random sublattices by Hislop, Kirsch and Krishna \cite{HKK}. %Hislop, Kirsch and Krishna also proved in the case of random sublattices the existence of absolutely continuous spectrum on the half-line $[0,\infty)$, however they were not able to rule out the existence of localized eigenfunctions for positive energies (i.e. embedded pure point spectrum).

\subsection{Results}
We consider the operator \eqref{random Schrodinger} on the torus $\T^2_L$ \footnote{
We have chosen periodic boundary conditions here. However, our proofs can easily be adapted to Neumann or Dirichlet conditions. This simply leads to a different character in the spectral expansion of the Green's functions.
} and the random displacement process $\omega_L=\{\omega_\xi \mid \xi\in \T^2_L\cap \Z^2\}$. There is no particular reason for the choice of a standard torus. Our results still hold for rectangular tori. %(with the sole exception that in $d=3$ the set $\scrL_3$, which we will introduce below, has to be replaced with a full density subset of $[0,\infty)$, similarly as in $d=2$).

Rigorously, the formal operator \eqref{random Schrodinger} is realized by self-adjoint extension theory, as explained above, leading to the family of operators $-\Delta_{\omega_L,U}$, where $U\in U(N)$ and $N=\#\omega_L$. The choice $U=e^{\i\varphi}\Id_N$, $\varphi\in(-\pi,\pi)$, corresponds to the formal operator $H_{\omega_L}$ with $\alpha\neq 0$. We denote the associated self-adjoint extension by $-\Delta_{\omega_L,\varphi}$.

The spectrum of the operator $-\Delta_{\omega_L,\varphi}$ on $\T^2_L$ is discrete and the density of eigenvalues increases with $L$, according to Weyl's law, proportional with the volume of $\T^2_L$. 
%Given $E>0$, we denote the new eigenfunction of $-\Delta_{\omega_L,U}$, whose eigenvalue $\scrE=\scrE(L)$ (a random variable which depends on $\omega_L$) is closest\footnote{If there are two of equal distance, take the larger one.} to $E$, by $\Psi_E^L$ (we will explain below that $\omega_L$-a.s. the dimension of the eigenspace is $1$). Note that in the limit as $L\to\infty$ we have $\scrE\to E$.
The eigenfunctions of the operator $-\Delta_{\omega_L,\varphi}$ are given by random superpositions of Green's functions 
$$G_E^L(x)=\sum_{\omega_\xi\in\omega_L}c_{\xi,\omega_L}G_E^L(x,\xi+\omega_\xi).$$
We fix a normalization of the coefficients $c_{\xi,\omega_L}$. In fact, $\omega_L$-a.s. the space of superposition vectors $(c_{\xi,\omega_L})_{\xi\in\Z^2\cap \T^2_L}$ is of dimension $1$. 
So upon choosing our normalization, the superposition vector is a.s. unique.
\begin{normal}
For the eigenfunctions $G_E^L$ we fix the normalization 
$$\sum_{\xi\in\Z^2\cap \T^2_L}|c_{\xi,\omega_L}|^2=1.$$
\end{normal}

Now if $L$ is large compared with the localization length we should be able to observe exponential localization for the eigenfunctions $G_E^L$. 

We now proceed to define localization for the random Schr\"odinger operator $-\Delta_{\omega_L,\varphi}$. Let us define the smoothed $L^2$-densities
$$\Phi_E^L(x)=\int_{\T^2_L}\chi(x'-x)|G_E^L(x')|^2 dx'$$
and $$\varphi_E^L(x)=\int_{\T^2_L}\chi(x'-x)|g_E^L(x')|^2 dx', \quad g_E^L=G_E^L/\|G_E^L\|_{L^2(\T^2_L)}.$$
where $\chi(x)=\tilde{\chi}(|x|)$ s. t. $\|\chi\|_1=1$, with $\tilde{\chi}\in C^\infty_c(\R_+)$ such that $\supp\tilde{\chi}=[0,1]$ with $\tilde{\chi}|_{[0,\tfrac{1}{2}]}=1$, and decreasing on $[\tfrac{1}{2},1]$.

We define localization in terms of the two point correlation density of an eigenfunction
\begin{defn}\label{locdefn}
Fix $L\gg 1$. Denote by $\omega_L$ the random displacement model on $\T^2_L$, as defined above. Let $\varphi\in(-\pi,\pi)$ and consider the operator $-\Delta_{\omega_L,\varphi}$. Let $\chi_R=R^{-d}\chi(\cdot/R)$.

We say that the operator $-\Delta_{\omega_L,\varphi}$ satisfies $f$-localization on an interval $I=[a,b]$ if $\forall R$ with $\frac{L}{10}\geq R\geq\frac{100}{\sqrt{E}}$ and $\forall x,y\in \T^2_L$ s. t.\footnote{This ensures that the balls we average over are sufficiently far apart.} $|x-y|\geq 4R$ we have a.s.
\begin{equation}\label{decay_cond}
\forall E\in[a,b]: \varphi_E^L(x)\varphi_E^L(y)
\leq \scrC_{\omega_L} f(|x-y|)
\end{equation}
where $\scrC_{\omega_L}$ is integrable on the sample space.

We say that $-\Delta_{\omega_L,U}$ satisifes exponential localization on $I$ if it satisfies $f$-localization with $f(x)=Ae^{-B|x|}$,
where $A,B$ are constants which may depend on the choice of $I$.
\end{defn}
\begin{remark}
We point out that, since the eigenfunction $G_E^L$ is a superposition of Green's functions, it is necessary to define the localization bound by integrating against a smooth test function in order to deal with the singularities.%(for $d=2$ and $|x|\ll 1/\sqrt{E}$ of order $\log(\sqrt{E}x)$ and for $d=3$ of order $(\sqrt{E}x)^{-1}$). 
It is, furthermore, important to ensure that the diameter of the support is large compared to the wavelength $1/\sqrt{E}$ in order to ensure that we are not integrating over a region which is entirely contained in the immediate vicinity of a singularity.
\end{remark}

\begin{cor}
Let $E=\inf\{E'\in[a,b]\}$. $f$-localization implies that $\exists C>0$ s. t. $\forall x,y\in \T^2_L, |x-y|\geq4R$ we have $$\E(\varphi_E^L(x)\Phi_E^L(y))\leq C f(|x-y|).$$
\end{cor}
\begin{proof}
From the definition of localization we have $$\varphi_E^L(x)\Phi_E^L(y)\leq \scrC_{\omega_L}\|G_E^L\|_2^2 f(|x-y|)$$
and the corollary follows by observing that $\E(C_{\omega_L}\|G_E^L\|_2^2)\leq \E(C_{\omega_L}^2)^{1/2}\E(\|G_E^L\|_2^4)^{1/2}$. As $L\to\infty$ the quantity $\E(\|G_E^L\|^4_2)$ remains bounded as it converges to $\E(\|G_E\|_2^4)$ (recall the normalization $\sum_{\xi}|c_\xi|^2=1$). Similarly $\E(C_{\omega_L}^2)$ converges to $\E(C_\omega^2)$ as $L\to\infty$.
\end{proof}
 
The following theorem is our main result.
\begin{thm}\label{main}
Fix $f:\R_+\mapsto\R_+$ a continuous, stricly decreasing function such that $f(x)=O(x^{-\alpha})$ for sufficiently large $\alpha>0$. Given $L\gg 1$ let $\omega_L$ denote the random displacement model on $\T^2_L$, as defined above. Let $\varphi\in(-\pi,\pi)$ and denote the associated self-adjoint extension by $-\Delta_{\omega_L,\varphi}$. 

Then there exists $E_0>0$ such that for any $b>a>E_0$ there exist $L_0\gg 1$, $\epsilon_0\ll 1$ s.t. for any $L\geq L_0$ the operator $-\Delta_{\omega_L,\varphi}$ cannot satisfy $f$-localization on $[a,b]$.
\end{thm}

An analogous result can easily be proved for the case $d=3$ following the exact same argument that is presented in this paper. Instead of the bounds on lattice point sums in $d=2$ from \cite{RU,U2} one uses the analogous bounds for $d=3$ which were proven in \cite{Y}.

\subsection*{Acknowledgements} 
This work was largely carried out as a Postdoc at the Institute of Theoretical Physics at CEA Saclay and completed while a Postdoc, supported by the Labex CEMPI (ANR-11-LABX-0007-01), at the Laboratoire Paul Painlev\'e, Universit\'e Lille 1. 

I would, in particular, like to thank St\'ephane Nonnenmacher for numerous discussions about this work and very useful suggestions which have contributed to the improvement of this paper. Furthermore, I would like to thank Fr\'ed\'eric Klopp for very useful discussions about Anderson localization and the problem of delocalization.

\section{Proof of Theorem \ref{main}}

Throughout this section and the rest of the paper expectation values are taken with respect to the random variable $\omega_L$. We will therefore omit the subscript.

The proof of Theorem \ref{main} works by contradiction. We assume $f$-localization of the operator $-\Delta_{\omega_L,\varphi}$ on an interval $[a,b]$ for sufficiently large $a\gg 1$, which ensures $\frac{100}{\sqrt{a}}<\tfrac{1}{4}$, and show that for a sufficiently small disorder parameter $\epsilon_0$ and sufficiently large decay rate $\alpha$ and torus size $L$ this assumption leads to a contradiction. 

Denote by $\scrS'$ the subsequence of density one of the set of integers representable as a sum of two squares as constructed in appendix \ref{construct} (cf. the subsequence $\scrS'$ in Thm 1.1, p. 3 in \cite{U2}). 

Let $L\gg 1$ and choose any $n_k\in\scrS'$ s.t. $[\frac{n_k}{L^2},\frac{n_k}{L^2}]\subset[a,b]$. Fix small $\delta_0>0$. We choose $$E=\inf\{E'\in\sigma(-\Delta_{\omega_L,\varphi})\cap[\frac{n_k+\delta_0}{L^2},\frac{n_{k+1}-\delta_0}{L^2}]\}.$$

We will make use of the following two propositions. 
The first proposition is proven in section \ref{lower poly}.
\begin{prop}\label{lowpo}
There exists $x_0\in \T^2_L$, $A>0$ such that $\E(\varphi_E^L(x_0))\gtrsim L^{-A}$.
\end{prop}

The second one is proven in section \ref{deloc}.
\begin{prop}\label{equithm}
Let $a\in C^\infty(\T^2)$, $\hat{a}(0)\neq0$. We have for any $x_0\in \T^2_L$ and $\psi_\lambda=\varphi_\E^L(x_0)^{1/2}g_\lambda$
$$\E(\left\langle a\psi_\lambda,\psi_\lambda\right\rangle)\gtrsim \hat{a}(0)\E(\varphi_E^L(x_0)).$$
\end{prop} 

So by Proposition \ref{lowpo} there exists $x_0\in \T^2_L$ s. t. for $\psi_\lambda=\phi_E^L(x_0)^{1/2}g_\lambda$, $\lambda=EL^2$, $a\in C^\infty(\T^2)$,
$$\E(\left\langle a\psi_\lambda,\psi_\lambda\right\rangle)\gtrsim \hat{a}(0)L^{-A}.$$
Fix $y\in \T^2_L$ s. t. $|x_0-y|\asymp L$, and take $a=\chi_\epsilon(\cdot-x)$, $x=y/L$. We obtain for $EL^2\gg 1$
\begin{equation}
\begin{split}
\phi_E^L(x_0)\phi_E^L(y)
&=\phi_E^L(x_0)\int_{\T^2_L}\chi_{\epsilon L}(y'-y)|g_E^L(y')|^2dy'\\
&=L^{-2}\phi_E^L(x_0)\int_{\T^2_L}\chi_{\epsilon}(\frac{y'-y}{L})|g_E^L(y')|^2dy'\\
&=L^{-2}\phi_E^L(x_0)\int_{\T^2}\chi_{\epsilon}(x'-x)|g_{EL^2}(x')|^2dx'
\end{split}
\end{equation}
where we used $\int_{\T^2}\chi_\epsilon(x)dx=1$ and $g_E^L(y)=L^{-1}g_{EL^2}(y/L)$.

We obtain $$L^{-2-A}\lesssim \E(\varphi_E^L(x_0)\varphi_E^L(y))\leq f(L)\lesssim L^{-\alpha}$$ which leads to a contradiction for $L\gg 1$ for $\alpha>2+A$.

\section{Background}\label{back}

\subsection{Self-adjoint extension theory}\label{extensions}
Let $x_1,\cdots,x_N$ be distinct points on $\T^2$. Denote $\x=\{x_1,\cdots,x_N\}$. This section will be concerned with the rigorous mathematical realization of the formal operator 
\begin{equation}\label{formal}
-\Delta+\sum_{j=1}^N \alpha_j \delta(x-x_j), \quad \alpha_1,\cdots,\alpha_N\in\R
\end{equation}
where $\Delta$ denotes the Laplacian with Dirichlet boundary conditions.

Define $D_{\x}:=C^\infty_c(\T^2-\x)$ and consider the restricted Laplacian $H=-\Delta|_{D_\x}$. Denote the Green's function of the Laplacian on $\T^2$ by $$G_\lambda(x,y)=(\Delta+\lambda)^{-1}\delta(x-y).$$

The operator $H$ has deficiency indices $(N,N)$ and the deficiency spaces are spanned by the bases of deficiency elements $\{G_{\pm\i}(x,x_1),\cdots,G_{\pm\i}(x,x_N)\}$ respectively. There exists a family of self-adjoint extensions of $H$ which is parameterized by the group $U(N)$. We denote the self-adjoint extension of $H$ associated with a matrix $U\in U(N)$ by $-\Delta_{\x,U}$. 

%We remark that a choice of a diagonal unitary matrix $D=\diag(e^{\i\varphi_1},\cdots,e^{\i\varphi_N}), \varphi_1,\cdots,\varphi_N\in(-\pi,\pi)$, corresponds to a realization of the formal operator \eqref{formal}, where there is no ``interaction'' between the individual potentials.

%For simplicity we only consider extensions corresponding to $D=e^{\i\varphi}I$, $\varphi\in(-\pi,\pi)$, which means we assume that all scatterers are of the same coupling strength and that there is no ``interaction'' between them. We denote the corresponding extension of $H$ by $-\Delta_\varphi$.

\subsubsection{Spectrum and eigenfunctions}\label{spectrum}
As explained above there are two types of eigenfunctions of the operator $-\Delta_{\x,U}$. Generic and non-generic eigenfunctions.

Our results hold for both types of new eigenfunctions. Since non-generic eigenfunctions only occur with probability $0$ and do not feel the presence of all impurities, we will ignore them for the rest of the paper, and focus on the generic eigenfunctions.

To find the new eigenfunctions of the operator $-\Delta_{\x,U}$ we want to solve
\begin{equation}
(\Delta_{\x,U}+\lambda)g_\lambda=0.
\end{equation}

We may write $g_\lambda$ in the decomposition
\begin{equation}\label{decomp}
g_\lambda=f_\lambda+\left\langle v,\G_\i\right\rangle + \left\langle Uv,\G_{-\i}\right\rangle
\end{equation} 
where $\G_{\lambda}(x)=(G_\lambda(x,x_1),\cdots,G_\lambda(x,x_N)),v\in\C^N$ and $g_\lambda\in C^\infty_c(\T^2)$.

So we have
\begin{equation}
(\Delta+\lambda)f_\lambda+(-\i+\lambda)\left\langle v,\G_\i\right\rangle + (\i+\lambda)\left\langle Uv,\G_{-\i}\right\rangle=0.
\end{equation}
We apply the resolvent $(\Delta+\lambda)^{-1}$, for $\lambda\not\in\sigma(-\Delta)$, and obtain
\begin{equation}
f_\lambda+\frac{-\i+\lambda}{\Delta+\lambda}\left\langle v,\G_\i\right\rangle + \frac{\i+\lambda}{\Delta+\lambda}\left\langle Uv,\G_{-\i}\right\rangle= 0
\end{equation}
By the repeated resolvent identity $$\frac{\mp\i+\lambda}{(\Delta+\lambda)(\Delta\pm\i)}=-\frac{1}{\Delta+\lambda}+\frac{1}{\Delta\pm\i}$$ we can rewrite this equation as
\begin{equation}\label{FNeq}
f_\lambda-\left\langle v,\G_\lambda-\G_\i\right\rangle - \left\langle Uv,
\G_\lambda-\G_{-\i}\right\rangle = 0
\end{equation}
Furthermore, note that we can write more compactly $$\left\langle v,\G_\lambda-\G_\i\right\rangle + \left\langle v,U^{-1}(\G_\lambda-\G_{-\i})\right\rangle
=\left\langle v, \A_\lambda \right\rangle$$
where $\A_\lambda(x)=(\G_\lambda-\G_\i)(x) + U^{-1}(\G_\lambda-\G_{-\i})(x)$. 

Now, since $f_\lambda=\left\langle v,\A_\lambda \right\rangle\in C^\infty_c(\T^2)$, we obtain the equations (set $x=x_k$ for $k=1,\cdots,N$)
\begin{equation}
\left\langle v, \A_\lambda(x_k) \right\rangle = 0, 
\quad k=1,\cdots, N,
\end{equation}
which we can rewrite as the matrix equation
\begin{equation}\label{speceqn}
M_\lambda\;v = 0
\end{equation}
where $F_\x(\lambda)=M_\lambda=(\A_\lambda(x_1),\cdots,\A_\lambda(x_N))$. 

So in order to find nontrivial solutions we need to solve the spectral equation 
\begin{equation}\label{spectral}
\det M_\lambda=0.
\end{equation}
We note that $\det M_\lambda$ is a meromorphic function of $\lambda$ with poles at the Laplacian eigenvalues, which we recall are given by the set $\scrS=\{n \mid n=4\pi^2(x_1^2+x_2^2), \;x_1,x_2\in\Z\}=\{0=n_0<n_1<n_2<\cdots\}$. 

Given a solution $\lambda\in\sigma(-\Delta_{\x,U})$ the corresponding eigenfunction will be given by
\begin{equation}
G^N_{\lambda,\x}(x)=\left\langle(\Id+U)v,\G_\lambda(x)\right\rangle=\sum_{j=1}^N d_{\lambda,j}(\x)G_\lambda(x,x_j), \quad v\in\ker M_\lambda
\end{equation}
which can be seen by substituting identity \eqref{FNeq} in \eqref{decomp}. 

Note that, for a full measure subset of $\x\in \T^{2N}$, we have that $d_{\lambda,j}(\x)\neq0$, $j=1,\cdots,N$, and $\dim\ker M_\lambda=1$. When we put a continuous probability measure on the space $\T^{2N}$, we may say that these statements hold with probability $1$.

\subsection{Scaling to the standard torus}\label{scaling}

It can easily be seen that the formal definition of the operator $-\Delta_{U,\x}$ via the theory of self-adjoint extensions corresponds to the standard Laplacian $-\Delta$ acting on functions $f\in C^\infty(\T^2-\x)$ where $\Delta f+c_1\delta_{x_1}+\cdots+c_N\delta_{x_N}\in L^2(\T^2)$, where $c_j\in\C$, $j=1,\cdots,N$, and $f$ diverges logarithmically at each of the points $x_j$, where the constants in the asymptotics depend on the choice of the matrix $U$.

Let $f\in L^2(\T^2_L)$ and define by $g(y)=f(Ly)$ a function $g\in L^2(\T^2)$. Let $L\x=(Lx_1,\cdots,Lx_N)$.
It can easily be seen that the eigenvalue problem $$(\Delta_{U,L\x}+E)f=0$$ on the large torus $\T^2_L$ corresponds to the eigenvalue problem $$(L^{-2}\Delta_{U,\x}+E)g=0$$ on the unit torus $\T^2$. If, in the first problem we study eigenfunctions with eigenvalue $E$ and the limit of large tori $L\to\infty$, then in the second problem this corresponds to studying the large eigenvalue limit $\lambda=E L^2\to\infty$.

\section{Proof of Proposition \ref{equithm}}\label{deloc}

We will study the spatial distribution properties of the wave functions of the formal operator 
\begin{equation}\label{Poisson T_L}
-\Delta+\alpha\sum_{\omega_\xi\in\omega_L}\delta(x-\xi-\omega_\xi), \quad \alpha\in\R
\end{equation}
on the torus $\T^2_L$, in the limit $L\to\infty$, where $N=\#\omega_L$ and $\omega_L$ denotes the random displacement model on $\T^2_L$. Rigorously, we proceed as above and realize the formal Hamiltonian as the self-adjoint extension of the restricted Laplacian $-\Delta|_{C^\infty_c(\T^2_L-\omega_L)}$.

Consider a fixed interval $[a,b]\subset\R_+$. Assume an eigenfunction with eigenvalue $E\in[a,b]$, as chosen above, of the random Schr\"odinger operator \eqref{random R2} is exponentially localized. This localization should also be observed on a sufficiently large torus $\T^2_L$, provided $L$ is much larger than the localization length (which depends on the choice of interval).

The eigenfunctions of the operator \eqref{Poisson T_L} are given by random superpositions of Green's functions (see also appendix \ref{Green scaling} for the definition and relation between the Green's functions on the tori $\T^2$ and $\T^2_L$)
$$G_E^L(x)=\sum_{\omega_\xi\in\omega_L}c_{\xi,\omega_L}G_E^L(x,\xi+\omega_\xi).$$

We recall from the end of subsection \ref{spectrum} that, almost surely, $\dim\ker M_\lambda=1$, where $\lambda=E L^2$.
\begin{normal}
We make the convention that the coefficients $c_{\xi,\omega_L}$ are normalized to ensure that $$\sum_{\xi\in \T^2_L\cap\Z^2}|c_{\xi,\omega_L}|^2=1.$$
\end{normal}

\subsection{Localization implies bounds on correlations of coefficients}
Fix any $x_0\in \T^2_L$. Let $d_{\xi,\omega_L}=\varphi_E^L(x_0)^{1/2}c_{\xi,\omega_L}$. We introduce the two point correlation density $$\Psi_E^L(y)=\varphi_E^L(x_0)^{1/2}G_E^L(y).$$
%where $g_E^L=G_E^L/\|G_E^L\|_2$, so the normalization of the coefficients $c_{\xi,\omega_L}$ is independent of the random process $\omega_L$.

We readily find that $f$-localization of the correlation density $\Psi_E^L$ implies for each $\xi\in\Z^2$ the bound $\E(|d_{\xi,\omega_L}|^2)\lesssim\footnote{For positive functions $f,g$ we denote by $f\lesssim g$ that there exists a constant $C>0$ such that $f\leq Cg$} f(|\xi-x_0|)$. 
So, at low disorder, any localization of the eigenfunctions of the operator $H_{\omega_L}$ really translates directly into a corresponding bound on the discrete correlation function $d_{\xi,\omega_L}:\Z^2\cap \T^2_L \mapsto \C$.

We have the following lemma.
\begin{lem}\label{discrete}
Let $I=[a,b]$ with\footnote{This ensures that $1/\sqrt{E}\ll 1$ and we can therefore pick $1/\sqrt{E}\ll R_0\ll 1$.} $a\gg 1$ and $f:\R_+\to\R_+$ be a strictly decreasing function, which may depend on the choice of the interval $I$. Assume that $-\Delta_{\varphi,\omega_L}$ is $f$-localized on $I$.

It follows that there exists a positive constant $C_{f}$ such that for our chosen $E\in I$
$$\forall \xi\in\Z^2\cap \T^2_L: \; \E(|d_{\xi,\omega_L}|^2)\leq C_{f} b^2 f(|\xi-x_0|).$$
\end{lem}
\begin{proof}
Fix $R_0$ s.t. $1/\sqrt{a}\ll R_0<1/4$ (which is possible because $a\gg 1$).

Denote $\chi_0=\chi_{R_0}$ and $\chi_{E}=(-\Delta-E)\chi_0$ and note that, for some constant $C_0>0$, $$|\Delta\chi_0(y)|\leq C_0\sum_{y_i\in\scrC}\chi_0(y-y_i)$$ where $\scrC\subset\T_L^2$ is a finite cover such that $$B(0,2R)\subset\bigcup_{y_i\in\scrC}B(y_i,R).$$ %To see that the above inequality holds, note that $\supp(\Delta\chi_0)=B(0,2R)\setminus B(0,R)$, and $\sum_{y_i\in\scrC_d}\chi_0(y-y_i)\geq 1$ for $y\in\supp(\Delta\chi)$.

%Then $E^{-1}\chi_E\in\scrT_R$, since $E>1$. Since $\supp\chi_E=\supp\chi$, the test function $\chi_E$ is also admissable to measure the localization of $\psi_E^L$.

We have $$|\chi_E(y)|\leq C_0\sum_{y_i\in\scrC}\chi(y-y_i)+b\chi(y)$$
which implies, where $C$ is a positive constant,
\begin{equation}
\begin{split}
&\E\left(\Big|\int_{\T^2_L}\chi_{E}(y-x)\Psi_E^L(y)dy \Big|^2\right) \\
\leq &\E\left\{\left(\int_{\T^2_L}|\chi_{E}(y-x)|^{1/2}|\chi_{E}(y-x)|^{1/2}|\Psi_E^L(y)|dy\right)^2\right\}\\
\leq &(C_0|\scrC|+b) \E\left(\int_{\T^2_L}|\chi_{E}(y-x)||\Psi_E^L(y)|^2 dy\right)\\
\leq&(C_0|\scrC|+b) \left(C_0\sum_{y_i\in\scrC}f(|x+y_i-x_0|)+bf(|x-x_0|)\right)\\
\leq & Cb^2 f(|x-x_0|).
\end{split}
\end{equation}

Now 
\begin{equation*}
\begin{split}
\int_{\T^2_L}\chi_{E}(y-x)\Psi_E^L(y)dy
=&\sum_{\omega_\xi\in\omega_L}d_{\xi,\omega_L} \int_{\T^2_L}\chi_{E}(y-x)G_E^L(y,\xi+\omega_\xi)dy\\
=&\sum_{\omega_\xi\in\omega_L}d_{\xi,\omega_L} \int_{\T^2_L}\chi_{E}(y)G_E^L(y,\xi+\omega_\xi-x)dy\\
=&\sum_{\omega_\xi\in\omega_L}d_{\xi,\omega_L} \chi_0(\xi+\omega_\xi-x)
\end{split}
\end{equation*}

If we now fix $x=\xi+\omega_\xi$, we get (recall that $4R_0<1$), where $C_{f}$ is a positive constant,
\begin{equation}
\begin{split}
\E(\underbrace{|\chi_0(0)|^2}_{=R_0^{-4}}|d_{\xi,\omega_L}|^2)
=&\E\left(\Big|\sum_{\omega_\eta\in\omega_L}d_{\xi,\omega_L}\chi(\eta+\omega_\eta-\xi-\omega_\xi)\Big|^2\right)\\
\leq& Cb^2 f(|\xi+\omega_\xi-x_0|)\\
\leq& \tilde{C}_{f}b^2 f(|\xi-x_0|)
\end{split}
\end{equation}
and $C_f=R_0^4\tilde{C}_{f}$.

\end{proof}

\subsection{Scaling}
Now we can identify the eigenfunction $G_E^L$ with the eigenfunction $G_\lambda$ with eigenvalue $\lambda=EL^2$ of the operator $-\Delta_{\Omega_L,U}$, where $\Omega_L=\{\xi/L+\omega_\xi/L\}_{\xi\in \T^2_L\cap\Z^2}$, on the standard torus $\T^2$: $G_\lambda(x)=G_E^L(Lx)$. 

We have the following spectral expansion for the Green's function on $\T^2$ with Dirichlet boundary conditions (cf. Appendix B, eq. \eqref{Greenexp}), which is valid for $\lambda\notin\sigma(-\Delta)$ in the distributional sense,
\begin{equation}
\begin{split}
G_\lambda(x,y)=&(-\Delta-\lambda)^{-1}\delta(x-y)=\sum_{\xi\in\Z^2}c_\lambda(\xi)e_{\xi}(x-y)\\ c_\lambda(\xi)=&\frac{1}{|\xi|^2-\lambda}
\end{split}
\end{equation}
where $e_\xi(x)=e^{2\pi\i\left\langle \xi,x \right\rangle}$.

Now recall $$\Psi_E^L(y)=\varphi_E^L(x_0)^{1/2}G_E^L(y)$$ and consider the rescaled correlation function $$\Psi_\lambda(x)=\varphi_E^L(x_0)^{1/2}G_\lambda(x).$$

From the observations above we conclude that $\Psi_\lambda$ is of the form (for the scaling of the Green's function cf. Appendix \ref{Green scaling})
\begin{equation}
\begin{split}
\Psi_\lambda(x)=&\sum_{\eta\in\Z^2\cap \T^2_L}d_{\eta,\omega_L} G_{\lambda}(x,\eta/L+\omega_\eta/L)\\
=&\sum_{\xi\in\Z^2}c_\lambda(\xi)D_{\omega_L}(\xi)e_{\xi}(x), \\
D_{\omega_L}(\xi)=&\sum_{\eta\in \Z^2\cap \T^2_L}d_{\eta,\omega_L}e_\xi(-\eta/L-\omega_\eta/L).
\end{split}
\end{equation}
%where $G_{\lambda}(x,y)=(-\Delta-\lambda)^{-1}\delta(x,y)$ denotes the Green's function on $B_1$ with Dirichlet boundary conditions, $e_\xi(x)=e^{2\pi\i\left\langle \xi,x \right\rangle}$.
% and the coefficients are the same, $d_{j,\x}=d_{j,\x_L}$.

\subsection{Approximation on thin annuli}
Let $\delta>0$ as in \cite{U2}.
We define $$\Psi_\lambda^{\delta}(x)=\sum_{\xi\in\Z^2\cap A(n_k,n_k^\delta)}c_\lambda(\xi)D_{\omega_L}(\xi)e_{\xi}(x)$$ and $$\Psi_\lambda^R(x)=\sum_{\xi\in\Z^2\cap A(n_k,n_k^{\delta})^c}c_\lambda(\xi)D_{\omega_L}(\xi)e_{\xi}(x).$$

We have, for $\zeta\in\Z^2$,
\begin{equation*}
\begin{split}
|\left\langle e_{\zeta} \Psi_\lambda,\Psi_\lambda\right\rangle|
\leq &|\left\langle e_{\zeta} \Psi_\lambda^\delta,\Psi_\lambda^\delta\right\rangle|
+\|\Psi_\lambda^R\|_2^2+2\|\Psi_\lambda^R\|_2\|\Psi_\lambda\|_2\\
\end{split}
\end{equation*}
where
\begin{equation*}
\begin{split}
|\left\langle e_{\zeta} \Psi_\lambda^\delta,\Psi_\lambda^\delta\right\rangle|
\leq&\sum_{\xi\in\Z^2\cap A(n_k,n_k^\delta)}|c_\lambda(\xi)D_{\omega_L}(\xi)c_\lambda(\xi+\zeta)D_{\omega_L}(\xi+\zeta)|\\
\end{split}
\end{equation*}
First of all we have 
\begin{equation*}
\begin{split}
&\sum_{\xi\in\Z^2\cap A(n_k,n_k^{\delta})}|c_\lambda(\xi)D_{\omega_L}(\xi)c_\lambda(\xi+\zeta)D_{\omega_L}(\xi+\zeta)|\\
\leq &\left(\sum_{\xi\in\Z^2\cap A(n_k,n_k^{\delta})}c_\lambda(\xi+\zeta)^2|D_{\omega_L}(\xi+\zeta)|^2\right)^{1/2}
\left(\sum_{\xi\in\Z^2\cap A(n_k,n_k^{\delta})}c_\lambda(\xi)^2|D_{\omega_L}(\xi)|^2\right)^{1/2}
\\
<\; &\|\Psi_\lambda\|_2
\Bigg(
\sum_{\substack{\xi\in\Z^2\cap A(n_k,n_k^{\delta})\\ |\xi|^2<n_k}}
c_{n_k}(\xi+\zeta)^2|D_{\omega_L}(\xi+\zeta)|^2 
+\sum_{\substack{\xi\in\Z^2\cap A(n_k,n_k^{\delta})\\ |\xi|^2>n_{k+1}}}
c_{n_{k+1}}(\xi+\zeta)^2|D_{\omega_L}(\xi+\zeta)|^2
\Bigg)^{1/2}
\end{split}
\end{equation*}
%where $\xi_0\in\Z^2$ is such that $|\xi_0|^2=n_{--}(\lambda)$.

And, secondly,
\begin{equation*}
\begin{split}
\|\Psi_\lambda^R\|_2^2
&\leq\sum_{\xi\in\Z^2\cap A(n_k,n_k^{\delta})^c}c_\lambda(\xi)^2|D_{\omega_L}(\xi)|^2\\
\leq &
\sum_{\substack{\xi\in\Z^2\cap A(n_k,n_k^{\delta})^c\\ |\xi|^2<n_k}}
c_{n_k}(\xi)^2|D_{\omega_L}(\xi)|^2
+\sum_{\substack{\xi\in\Z^2\cap A(n_k,n_k^{\delta})^c\\ |\xi|^2>n_{k+1}}}
c_{n_{k+1}}(\xi)^2|D_{\omega_L}(\xi)|^2
\end{split}
\end{equation*}

%Because of the decay of the $d_{j,\x_L}$ we have that there exists a finite set $\Omega$ such that
%\begin{equation}
%\sum_{x_j\in\Omega}|d_{j,\x_L}|\geq \beta\sum_{x_j\in\x_L}|d_{j,\x_L}|, \quad \beta\in(0,1).
%\end{equation}

Because of the decay of the $\E(|d_{\eta,\omega_L}|^2)$ we have, for $N_0$ large, that there exists a constant $C_0>0$ such that we have the following bound
\begin{equation*}\label{ineq1}
\begin{split}
\E(|D_{\omega_L}(\xi+\zeta)|^2)\leq \E\left\{(\sum_{\eta\in \T^2_L\cap\Z^2}|d_{\eta,\omega_L}|)^2\right\}
&=\sum_{\eta_1,\eta_2\in \T^2_L\cap\Z^2}\E(|d_{\eta_1,\omega_L}d_{\eta_2,\omega_L}|)\\
&\leq \sum_{\eta_1,\eta_2\in \T^2_L\cap\Z^2}\E(|d_{\eta_1,\omega_L}|^2)^{1/2}
\E(|d_{\eta_2,\omega_L}|^2)^{1/2}\\
&\leq C_0\sum_{\substack{\eta_1,\eta_2\in \T^2_L\cap\Z^2 \\ |\eta_1-x_0|,|\eta_2-x_0|\leq N_0}}
\E(|d_{\eta_1,\omega_L}|^2)^{1/2}\E(|d_{\eta_2,\omega_L}|^2)^{1/2}\\
&= C_0\left(\sum_{\substack{\eta\in \T^2_L\cap\Z^2 \\ |\eta-x_0|\leq N_0}}\E(|d_{\eta,\omega_L}|^2)^{1/2}\right)^2\\
&\leq C_0 N_0\sum_{\substack{\eta\in \T^2_L\cap\Z^2 \\ |\eta-x_0|\leq N_0}}\E(|d_{\eta,\omega_L}|^2)
\end{split}
\end{equation*}
where we picked $N_0$ large enough to ensure that (recall $f(r)=O(r^{-\alpha})$ and take $\alpha>4$)
\begin{equation}\label{ineq2}
\begin{split}
\sum_{\substack{\eta\in \T^2_L\cap\Z^2 \\ |\eta-x_0|> N_0}}\E(|d_{\eta,\omega_L}|^2)^{1/2}\lesssim b\int_{|x-x_0|> N_0}f(|x-x_0|)^{1/2}dx
&=b\int_{N_0}^{\infty}f(r)^{1/2}r dr\\
&\lesssim_{d,f} bN_0^{-\alpha/2+2}\\
&\leq \frac{1}{2}L^{-A}\\
&\leq \frac{1}{2}\E(\varphi_E^L(x_0))=\frac{1}{2}\sum_{\eta\in \T^2_L\cap\Z^2}\E(|d_{\eta,\omega_L}|^2)\\
&\leq \frac{1}{2}\sum_{\eta\in \T^2_L\cap\Z^2}\E(|d_{\eta,\omega_L}|^2)^{1/2},
\end{split}
\end{equation}
say $N_0=1000\, (bL^A)^{\frac{1}{\frac{\alpha}{2}-2}}$, and we used $\E(|d_{\eta,\omega_L}|^2)\leq \E(\varphi_E^L(x_0))\leq 1$.

Let $a\in C^\infty(\T^2)$ with Fourier expansion 
$$a(x)=\sum_{\zeta\in\Z^2}\hat{a}(\zeta)e_{\zeta}(x).$$

It suffices to prove the result for any trigonometric polynomial of degree $J$, a standard approximation argument then yields the result for $C^\infty$ test functions (see for instance \cite{U2}).

Let $\tilde{a}=a-\hat{a}(0)$. We have
$$|\left\langle \tilde{a} \Psi_\lambda,\Psi_\lambda\right\rangle|
\leq |\left\langle \tilde{a} \Psi_\lambda^\delta,\Psi_\lambda^\delta\right\rangle|
+\|\tilde{a}\|_\infty \|\Psi_\lambda^R\|_2^2
+\|\tilde{a}\|_\infty \|\Psi_\lambda^R\|_2 \|\Psi_\lambda\|_2.$$

Now 
$$|\left\langle \tilde{a} \Psi_\lambda^\delta,\Psi_\lambda^\delta\right\rangle|
\leq \sum_{\zeta\in\Z^2\setminus\{0\}, |\zeta|\leq J}|\hat{a}(\zeta)|
|\left\langle e_\zeta \Psi_\lambda^\delta,\Psi_\lambda^\delta\right\rangle|
\leq \sum_{\zeta\in\Z^2\setminus\{0\}, |\zeta|\leq J}|\hat{a}(\zeta)|\scrA_{\delta,\zeta}(\omega_L)
:=\scrA(\omega_L)
$$
where, for $\zeta\neq 0$ and $\lambda=EL^2\in[n_k+\delta_0,n_{k+1}-\delta_0]\subset(n_k,n_{k+1})$,
\begin{equation}
\begin{split}
\E(\scrA_{\delta,\zeta}(\omega_L))
\leq&\E(\|\Psi_\lambda\|_2^2)^{1/2}\left(\sum_{\substack{\xi\in A(n_k,n_k^\delta)\\|\xi|^2<n_k}}|c_{n_k}(\xi+\zeta)|^2
+\sum_{\substack{\xi\in A(n_k,n_k^\delta)\\|\xi|^2>n_{k+1}}}|c_{n_{k+1}}(\xi+\zeta)|^2\right)^{1/2}\\
&\hspace{2.2in}\times \sqrt{C_0N_0}\left(\sum_{\substack{\eta_j\in \T^2_L\cap\Z^2\\ |\eta-x_0|\leq N_0}}\E(|d_{\eta_j,\omega_L}|^2)\right)^{1/2}\\
&\lesssim_\epsilon N_0\lambda^{-\delta_1+\epsilon}\E(\varphi_E^L(x_0))
\end{split}
\end{equation}
for some $\delta_1>0$, where we used a bound on sums over lattice points in shifted annuli which, for $d=2$, are given in \cite{RU} (see also appendix \ref{construct}).

We also used
the normalization of the $c_\eta$, i.e. $$\sum_{\eta\in \Z^2\cap \T^2_L} |d_\eta|^2=\varphi_E^L(x_0)\sum_{\eta\in \Z^2\cap \T^2_L} |c_\eta|^2=\varphi_E^L(x_0),$$
and the following lemma.
\begin{lem}
For $\lambda\in [n_k+\delta_0,n_{k-1}-\delta_0]$ we have the bound $$\E(\|\Psi_\lambda\|_2^2)\lesssim_\epsilon N_0\delta_0^{-2}\lambda^\epsilon\E(\varphi_E^L(x_0)).$$
\end{lem}
\begin{proof}
We have the inequality
\begin{equation}
\begin{split}
\E(\|\Psi_\lambda\|_2^2)&=\sum_{\xi\in\Z^2}\E(c_\lambda(\xi)^2|D_{\omega_L}(\xi)|^2)\\
&\leq \sum_{|\xi|^2<n_k}c_{n_k}(\xi)^2\E(|D_{\omega_L}(\xi)|^2)+\sum_{|\xi|^2>n_{k+1}}c_{n_{k+1}}(\xi)^2\E(|D_{\omega_L}(\xi)|^2)\\
&\quad+\delta_0^{-2}\sum_{|\xi|^2=n_k,n_{k+1}}\E(|D_{\omega_L}(\xi)|^2)\\
&\lesssim_{\epsilon} N_0\E(\varphi_E^L(x_0))\left(\delta_0^{-2}n_k^{\epsilon}+\sum_{n<n_k}\frac{n^\epsilon}{(n-n_k)^2}
+\sum_{n>n_{k+1}}\frac{n^\epsilon}{(n-n_{k+1})^2}\right)\\
&\lesssim_\epsilon N_0\delta_0^{-2}n_k^\epsilon \E(\varphi_E^L(x_0))
\end{split}
\end{equation}
where we used $r_2(n)\lesssim_\epsilon n^\epsilon$.
\end{proof}

We thus have
\begin{equation}
\E(|\left\langle a \Psi_\lambda^\delta,\Psi_\lambda^\delta\right\rangle|)
\leq\sum_{\zeta\in\Z^2\setminus\{0\}, |\zeta|\leq J}|\hat{a}(\zeta)|\E(\scrA_{\delta,\zeta}(\omega_L))
\lesssim_\epsilon \|\hat{a}\|_{l^1}N_0\lambda^{-\delta_1+\epsilon}\E(\varphi_E^L(x_0))
\end{equation}

Furthermore, we have
\begin{equation}
\begin{split}
\E(\|\Psi_\lambda^R\|_2^2)\leq\left(\sum_{\substack{\xi\in A(n_k,n_k^\delta)^c\\|\xi|^2<n_k}}c_{n_k}(\xi)^2
+\sum_{\substack{\xi\in A(n_k,n_k^\delta)^c\\|\xi|^2>n_{k+1}}}c_{n_{k+1}}(\xi)^2\right)
\times C_0N_0\sum_{\substack{\eta_j\in \T^2_L\cap\Z^2\\ j\leq N_0}}\E(|d_{\eta_j,\omega_L}|^2)\\
\lesssim_\epsilon N_0\lambda^{-\delta_2+\epsilon}\E(\varphi_E^L(x_0)).
\end{split}
\end{equation}

Therefore,
\begin{equation}
\begin{split}
&\E|\left\langle \tilde{a} \Psi_\lambda,\Psi_\lambda\right\rangle|\\
&\leq \E|\left\langle \tilde{a} \Psi^\delta_\lambda,\Psi^\delta_\lambda\right\rangle|+\|\hat{a}\|_{l^1}\E(\|\Psi^R_\lambda\|_2^2)
+\|\hat{a}\|_{l^1}\E(\|\Psi^R_\lambda\|_2^2)^{1/2}\E(\|\Psi^\delta_\lambda\|_2^2)^{1/2}\\
&\lesssim\|\hat{a}\|_{l^1}N_0\lambda^{-\delta}\E(\varphi_E^L(x_0))
\end{split}
\end{equation}
for some $\delta>0$ (again this follows from the bounds on lattice point sums of the above type which, for $d=2$, are given in \cite{RU} (see also appendix \ref{construct})). 

So we have (recall $\Psi_\lambda=\varphi_E^L(x_0)^{1/2}G_\lambda$ and $\psi_\lambda=\varphi_E^L(x_0)^{1/2}g_\lambda$)
\begin{equation}
\E(\left\langle a\psi_\lambda,\psi_\lambda\right\rangle)
=\hat{a}(0)\E(\varphi_E^L(x_0))+\E(\left\langle\tilde{a}\psi_\lambda,\psi_\lambda\right\rangle)
\end{equation}
and 
\begin{equation}
\begin{split}
\E(\left\langle\tilde{a}\psi_\lambda,\psi_\lambda\right\rangle)
\leq &\E(\left\langle\tilde{a}\Psi_\lambda,\Psi_\lambda\right\rangle)^{1/2}
\E(\left\langle\tilde{a}\Psi_\lambda,\Psi_\lambda\right\rangle \|G_\lambda\|_2^{-4})^{1/2} \\
\lesssim &\|\hat{a}\|_{l_1}N_0^{1/2}\lambda^{-\delta/2}\E(\varphi_E^L(x_0))^{1/2}\E(\varphi_E^L(x_0)\|G_\lambda\|_2^{-2})^{1/2}
\end{split}
\end{equation}
where we used for the second term the inequality
$$\E(\left\langle\tilde{a}\Psi_\lambda,\Psi_\lambda\right\rangle \|G_\lambda\|_2^{-4})
\leq \|\hat{a}\|_{l_1}\E(\varphi_E^L(x_0)\|G_\lambda\|_2^{-2})$$
in view of $|\left\langle\tilde{a}\Psi_\lambda,\Psi_\lambda\right\rangle|\leq \|\hat{a}\|_{l_1}\|\Psi_\lambda\|_2^2$.

Now we observe
$$\|G_\lambda\|_2^2=\sum_{\xi\in\Z^2}c_\lambda(\xi)^2|C_{\omega_L}(\xi)|^2\geq_\epsilon \lambda^{-\epsilon}\sum_{|\xi|^2=n_{k-1}}|C_{\omega_L}(\xi)|^2:=\lambda^{-\epsilon}F_{\omega_L}(n_{k-1})$$
where  
$$C_{\omega_L}(\xi)=\sum_{\eta\in \Z^2\cap \T^2_L}c_{\eta,\omega_L}e_\xi(-\eta/L-\omega_\eta/L)$$ and note $n_{k+1}-n_k\lesssim_\epsilon\ n_k^\epsilon$.

It follows for $\lambda=EL^2\gg 1$, $\hat{a}(0)\neq0$ and $\alpha>2A/\delta+4$ (recall $N_0=1000\, (bL^A)^{\frac{1}{\frac{\alpha}{2}-2}}$)
\begin{equation}
\begin{split}
\E(\left\langle a\psi_\lambda,\psi_\lambda\right\rangle)
=&\hat{a}(0)\E(\varphi_E^L(x_0))\\
&+\scrO_\epsilon(\|\hat{a}\|_{l_1}\lambda^{-\delta/2+\epsilon})\E(\varphi_E^L(x_0))^{1/2}\E(\varphi_E^L(x_0)\|G_\lambda\|_2^{-2})^{1/2}\\
=&\hat{a}(0)\E(\varphi_E^L(x_0))^{1/2}\\
&\times\E\left(\varphi_E^L(x_0)\left(1+\scrO(\frac{\|\hat{a}\|_{l_1}}{\hat{a}(0)^2}N_0^{1/2}\lambda^{-\delta/2+\epsilon})F_{\omega_L}(n_{k-1})^{-1}\right)\right)^{1/2}\\
\gtrsim& \hat{a}(0)\E(\varphi_E^L(x_0))
\end{split}
\end{equation} 
where we used that for $\diam\supp P_{\epsilon_0}=\epsilon_0\ll \frac{1}{\sqrt{b}}$ (weak disorder condition) we have
$$F_{\omega_L}(n_{k-1})\asymp r_2(n_{k-1})\int_{|v|^2=\sqrt{E}}|\hat{c}_{\omega_L}(v)|^2 d\theta 
\gtrsim (\log\lambda)^{\log 2/2 -\epsilon} \int_{|v|^2=\sqrt{E}}|\hat{c}_{\omega}(v)|^2 d\theta$$
as $L\to\infty$, where $d\theta$ denotes normalized Lebesgue measure on the circle of radius $\sqrt{E}$. 

To see this, we use the property of $n_k\in \scrS'$, that the lattice points $\xi/|\xi|$, $|\xi|^2=n_{k-1}$ equidistribute on $S^1$ as $k\to\infty$. Since for $|\xi|^2=n_{k-1}$ we have $|\xi/L|^2\sim E$, as $L\to\infty$, it follows that the lattice points $\xi/L$, $|\xi|^2=n_k$, equidistribute on the circle of radius $\sqrt{E}$. 

Furthermore, note that the assumption $\epsilon_0\ll \frac{1}{\sqrt{b}}$ implies
$|C_{\omega_L}(\xi)|^2\asymp |\hat{c}_{\omega_L}(\xi/L)|^2$, as $L\to\infty$, and $\hat{c}_{\omega_L}$ denotes the discrete Fourier transform 
$$\hat{c}_{\omega_L}(v)=\sum_{\eta\in\Z^2\cap \T^2_L}c_{\eta,\omega_L}e_\eta(-v)$$
of the function $c_{\omega_L}:\Z^2\cap \T^2_L\mapsto \C$, which converges to $$\hat{c}_\omega(v)=\sum_{\eta\in\Z^2}c_{\eta,\omega}e_\eta(-v), \quad \text{as $L\to\infty$}.$$

\section{Proof of Proposition \ref{lowpo}}\label{lower poly}

The proof is exactly analogous to the proof of Proposition \ref{equithm}, where $\varphi_E^L(x_0)$ is replaced with $1$.

Fix any $x_1\in \T^2_L$, and assume for a contradiction that for $\beta>2d$ and for any $\tfrac{1}{4}\leq R\leq \epsilon L$, $y\in \T^2_L$ we have
$$\E(\int_{\T^2_L} \chi_R(y'-y)|g_E^L(y')|^2dy')\leq (1+|x_1-y|)^{-\beta}.$$

We then have the following lemma whose proof is exactly analogous to that of Lemma \ref{discrete}.
\begin{lem}
There exists a positive constant $C$ s. t.
$$\forall \xi\in\Z^d\cap \T^2_L: \; \E(|c_{\xi,\omega_L}|^2)\leq C b^2 (1+|x_1-\xi|)^{-\beta}.$$
\end{lem}
We then apply this lemma to obtain the bound
$$\E(|C_{\omega_L}(\xi)|^2)\lesssim N_1$$
where we used the analogous estimate as in eq. \eqref{ineq1} as well as the normalization $\sum_\eta |c_{\eta,\omega_L}|^2=1$. We have to choose $N_1$ large enough such that
\begin{equation}\label{ineq2}
\begin{split}
\sum_{\substack{\eta\in \T^2_L\cap\Z^2 \\ |\eta-x_0|> N_0}}\E(|c_{\eta,\omega_L}|^2)^{1/2}\lesssim b\int_{|x-x_1|> N_1}(1+|x-x_1|)^{-\beta/2}dx
&=b\int_{N_1}^{\infty}(1+r)^{-\beta/2}rdr\\
&\lesssim_{d,f} bN_1^{-\beta/2+2}\\
&\leq \frac{1}{2}\\
&\leq\frac{1}{2}\sum_{\eta\in \T^2_L\cap\Z^2}\E(|c_{\eta,\omega_L}|^2)^{1/2},
\end{split}
\end{equation}
say $N_1=1000\, b^{\frac{1}{\frac{\beta}{2}-2}}$, and we used $\E(|c_{\eta,\omega_L}|^2)\leq 1$.

We readily derive for any $a\in C^\infty(\T^2)$ and $\lambda\gg1$, following the proof of Prop. \ref{equithm}, where we replace the factor $\varphi_E^L(x_0)$ with $1$, $$\left\langle a g_\lambda,g_\lambda\right\rangle\gtrsim \hat{a}(0).$$

Now, by our assumption, we have that for $y\in \T^2_L$ s.t. $|x_1-y|\asymp L$
\begin{equation}
\begin{split}
L^{-\beta}\asymp(1+|x_1-y|)^{-\beta}&\geq\int_{\T^2_L}\chi_{\epsilon L}(y'-y)|g_E^L(y')|^2dy'\\
&=L^{-2}\int_{\T^2_L}\chi_{\epsilon}(\frac{y'-y}{L})|g_E^L(y')|^2dy'\\
&=L^{-2}\int_{\T^2}\chi_{\epsilon}(x'-x)|g_{EL^2}(x')|^2dx'\\
&\gtrsim L^{-2}
\end{split}
\end{equation}
which leads to a contradiction.

Since our assumption is false, it follows that there exists $\tfrac{1}{4}\leq R'\leq \epsilon L$, $y=x_0\in \T^2_L$ s. t.
$$\E(\int_{\T^2_L} \chi_{R'}(y'-x_0)|g_E^L(y')|^2dy')> (1+|x_1-x_0|)^{-\beta}\geq (1+L)^{-\beta}.$$
And we have $\chi_{\epsilon L}\geq (\frac{R'}{\epsilon L})^2\chi_{R'}\gtrsim L^{-2}\chi_{R'}$, which implies
$$\E(\varphi_E^L(x_0))=\E(\int_{\T^2_L} \chi_{\epsilon L}(y'-x_0)|g_E^L(y')|^2dy')\gtrsim L^{-\beta-2}.$$

\begin{appendix}

\section{Constructing the subsequence $\scrS'$}\label{construct}

Denote by $\scrS=\{n \mid n=4\pi^2(x_1^2+x_2^2), x_1, x_2\in\Z\}=\{0=n_0<n_1<n_2<\cdots\}$ the set of Laplacian eigenvalues on the unit square $B=[-1/2,1/2]^2$ with Dirichlet boundary conditions, where we ignore multiplicities. 

There exists a subsequence $\scrS^*\subset \scrS$ of density $1$ and $\delta_2>0$ such that for all $n_k\in \scrS^*$:
\begin{itemize}
\item[(i)] 
%\begin{equation*}
$n_{k+1}-n_{k-1} \leq
%\begin{cases}
C_\epsilon n_k^\epsilon$%, %\;\text{if}\quad d=2%\\
%\\
%C, \;\text{if}\quad d=3.
%\end{cases}
%\end{equation*}
\item[(ii)] $\forall\lambda\in(n_k,n_{k+1})\; \forall\zeta\neq0,|\zeta|\leq J\; \forall\xi\in A(n_k,n_k^{\delta_2})\cap\Z^2: |c_\lambda(\xi+\zeta)|\lesssim\lambda^{-\delta_2}$
\item[(iii)] The lattice points on the circle $|\xi|^2=n_{k-1}$ become equidistributed as $k\to\infty$. 

This means that for any $g\in C^0(S^1)$ $$\frac{1}{r_2(n_{k-1})}\sum_{|\xi|^2=n_{k-1}} g\left(\frac{\xi}{|\xi|}\right)\rightarrow \int_{S^1}g(\theta)d\theta$$ as $k\to\infty$.
\end{itemize}

\begin{proof}
(i): To see this, recall that the elements of $\scrS$, integers representable as sums of $2$ squares, have mean spacing of order $\sqrt{\log n_k}$. %if $d=2$ and mean spacing of order $1$ if $d=3$. 
Therefore, the subsequence of $n_k$ s. t. $n_{k+1}-n_k\leq C_\epsilon n_k^\epsilon$ and those $n_k$ s. t. $n_k-n_{k-1}\leq C_\epsilon n_k^\epsilon$ are of density $1$ respectively. 
Consequently, their intersection is a subsequence of density $1$. %The same argument applies to the case $d=3$ and the subsequence of $n_k$ s. t. $n_{k+1}-n_k\leq C$ and those $n_k$ s. t. $n_k-n_{k-1}\leq C$.

(ii): This proof is exactly identical to the construction in sections 6 and 7 of \cite{RU}. Note the additional factor $4\pi^2$ which is due to the fact that we consider the standard torus $\R^2/\Z^2$ rather than the scaled torus $\R^2/2\pi\Z^2$ considered in \cite{RU}. %This argument can easily be extended for $d=3$, as shown in \cite{Y}.

(iii): It follows from classical equidistribution theorems that on a generic circle of radius $\sqrt{n_{k-1}}$ the lattice points $\xi\in\Z^d$ satisfying $|\xi|^2=n_k$ become equidistributed. Hence we may construct a density one subsequence of $n_k$ such that this holds for the neighbouring circles $|\xi|^2=n_{k-1}$.
\end{proof}

\section{Scaling of Green's functions}\label{Green scaling}

Here we simply point out the simple relationship between the Green's functions on the tori $\T^2$ and $\T^2_L$.

The Green's function on $\T^2$ is given by $G_\lambda=(-\Delta_{\T^2}-\lambda)^{-1}\delta(x-y)$ and, we recall, has the following Fourier expansion, which is convergent in the $L^2$-sense:
\begin{equation}\label{Greenexp}
G_\lambda(x,y)=\sum_{\xi\in\N^2}\frac{e_\xi(x-y)}{|\xi|^2-\lambda}.
\end{equation}

The Green's function on $\T^2_L$ is given by $G_E^L=(-\Delta_{\T^2_L}-E)^{-1}\delta(x-y)$ and via its Fourier expansion it can easily be related to the Green's function on $\T^2$ by scaling
\begin{equation}
\begin{split}
G_E^L(x,y)=&L^{-2}\sum_{\xi'\in L^{-1}\Z^2}\frac{e_{\xi'}(x-y)}{|\xi'|^2-E}\\
=&\sum_{\xi\in\Z^2}\frac{e_{\xi}((x-y)/L)}{|\xi|^2-E L^2}\\
=&G_\lambda\left(\frac{x}{L},\frac{y}{L}\right),\quad \lambda=E L^2.
\end{split}
\end{equation}

\end{appendix}

\end{document}